\documentclass{article}
\usepackage[utf8]{inputenc}
\usepackage{amsmath}
\usepackage{amsthm}
\usepackage{amssymb}
\usepackage{graphicx}
\usepackage{authblk}
\usepackage{bbm}
\usepackage{algorithm}
\usepackage[authoryear,round]{natbib}
\usepackage{caption}
\usepackage{subcaption}
\usepackage{booktabs}
\usepackage{hyperref}
\newtheorem{theorem}{Theorem}
\newtheorem{conjecture}{Conjecture}
\newtheorem{definition}{Definition}
\newtheorem{corollary}{Corollary}
\DeclareMathOperator{\poi}{Poi}
\DeclareMathOperator{\bin}{Bin}

\usepackage{fullpage}

\begin{document}

\title{Resampling-free bootstrap inference for quantiles\footnote{Acknowledgement: The authors gratefully acknowledge help and feedback from Anton Muratov, Shaobo Jin, Thommy Perlinger and Claire Detilleux.}}
\date{\today}
\author{Mårten Schultzberg\thanks{\href{mailto:mschultzberg@spotify.com}{mschultzberg@spotify.com}}~}
\author{Sebastian Ankargren\thanks{\href{mailto:sebastiana@spotify.com}{sebastiana@spotify.com}}}
\affil{Spotify, Experimentation Platform Team}
\maketitle

\begin{abstract}
Bootstrap inference is a powerful tool for obtaining robust inference for quantiles and difference-in-quantiles estimators. The computationally intensive nature of bootstrap inference has made it infeasible in large-scale experiments. In this paper, the theoretical properties of the Poisson bootstrap algorithm and quantile estimators are used to derive alternative resampling-free algorithms for Poisson bootstrap inference that reduce the computational complexity substantially without additional assumptions. These findings are connected to existing literature on analytical confidence intervals for quantiles based on order statistics. The results unlock bootstrap inference for difference-in-quantiles for almost arbitrarily large samples. At Spotify, we can now easily calculate bootstrap confidence intervals for quantiles and difference-in-quantiles in A/B tests with hundreds of millions of observations. 
\end{abstract}

\section{Introduction}
The use of randomized experiments in product development has seen an enormous increase in popularity over the last decade. Modern tech companies now view experimentation, often called A/B testing, as fundamental and have tightly integrated practices around it into their product development. The vast majority of A/B testing compares two groups, treatment and control, with respect to average treatment effects through calculation of difference-in-means. These comparisons are operationalized through standard $z$-tests that are simple to perform. With the rise of A/B testing, tests that do not compare average effects are also gaining more and more interest. Difference-in-quantiles, where treatment and control quantiles are compared, is one such test, where reasons for it might be that effects are not expected, or that they are difficult to identify, on average. For example, a change could be targeting users experiencing the largest amount of buffering, which means a difference-in-quantiles comparison for, say, the 90th percentile may be of more interest than the average buffering amount experienced by users. These tests are, however, much more difficult to perform, with non-standard sampling distributions, that severely complicate implementations.

In randomized experiments, a common technique for doing inference for estimators with non-standard or intractable sampling distributions is bootstrap \citep{efron1979}. Bootstrap is a resampling-based method where the sampling distribution of an estimator is estimated by resampling with replacement from the observed sample. Bootstrap inference is known to be consistent for quantile estimators and difference-in-quantiles estimators \citep{ghosh1984, falk1989} under mild conditions on the outcome distribution. The computational intensive nature of bootstrap has made the primary use case small-sample experiments. The large-scale online experiments run by tech companies often involve millions or even hundreds of millions of users. 

Recently, two prominent approaches for bootstrapping with big data have been proposed. Both of these methods are focused on finding computationally efficient implementations of the bootstrap approach rather than reducing its complexity. The first approach, known as the Poisson bootstrap, utilizes that bootstrap samples, i.e., multinomial sampling from the original sample, can be well approximated by Poisson frequencies \citep{hanley2006}. For example, \cite{Chamandy2012} showed that the Poisson bootstrap can be implemented in the MapReduce framework \citep{dean2008}, which enables powerful parallelization on clusters of computers. \cite{Chamandy2012} implemented a non-parametric bootstrap in MapReduce for linear estimators \citep{rao1973}, like means and sums or smooth functions thereof. Moreover, \cite{Chamandy2012} showed that semi-parametric estimators of non-linear estimators such as quantiles can be estimated using similar implementations. As an alternative to bootstrapping quantiles when samples are dependent, \cite{Liu2019} developed a method based on asymptotic arguments.

Another recent approach to bootstrapping that enables efficient implementations for big data is the so-called 'Bag of little bootstraps' \citep{kleiner2014}. This approach splits the full-sample inference problem into several smaller inference problem, and then weights the results together in a consistent manner. Bag of little bootstraps is non-parametric in the sense that it applies to any estimator, while still allowing efficient parallelization implementations. 

Both the Poisson bootstrap and the Bag of little bootstrap enable big data bootstrap inference for quantiles through parallelization. In other words, the computational complexity is overcome by efficient and scalable computing. The general complexity of the Poisson bootstrap algorithm is of the order of $O(CB)$, where $C$ is the complexity of the estimator calculated in each bootstrap sample (see, e.g., \citealp{Cormen2009} for an introduction to algorithmic complexity). For quantiles, most common estimators have complexity of order $O(N)$. This follows since most quantile estimators are based on the order statistics for which the complexity is of order $O(N)$ \citep{Cormen2009}, leading to $O(NB)$ for quantile bootstrap. 

Analytical confidence intervals for quantiles have been studied for a long time. For the one-sample quantile case, simple exact and distribution-free confidence intervals for population quantiles can be constructed using only order statistics, see for example \citet[p. 159]{gibbons2014}. These approaches unlock one-sample confidence intervals for massive samples, but they are largely absent in experimentation in the tech industry. A likely reason is that these approaches do not directly extend to the two-sample difference-in-quantile case.

The focus of this paper is to reduce the complexity of  Poisson confidence intervals (CIs) for quantiles and difference-in-quantiles type estimators. Specifically, the properties of the Poisson distribution and quantile estimators based on order statistics are used to simplify the problem. We note that for one-sample problems our approach exactly reproduces the exact order-statistic-based confidence intervals such as those described by e.g. \citet[~p.159]{gibbons2014}. We then show that our approach, as opposed to the exact CIs, can be easily extended to simplify the two-sample confidence interval problem. Specifically it is shown that, for two-sample difference-in-quantile estimators, it is sufficient to sample order statistics from the original samples according to a known probability distribution. This implies that it is sufficient to order the sample once, and then sample order statistics from the ordered sample. This reduces the computational complexity from $O(NB)$ to $ O(\max(N\log(N),B))$. These findings let us perform non-parametric Poisson bootstrap inference for difference-in-quantiles estimators for almost arbitrarily large data sets.

The rest of this paper is structures as follows. Section \ref{sec:AOS-CI}, gives a brief introduction to analytical exact CIs based on order statistics. Section \ref{sec:intro_poi_boot} gives a short overview of traditional Poisson bootstrap for quantiles and difference-in-quantiles. Section \ref{sec:resampling_free_boot} introduces the first contribution of the paper, resampling-free Poisson bootstrap for quantiles. Section \ref{sec:two-sample-extension} presents the extension to the difference-in-quantiles CIs together with Monte Carlo evidence of the coverage. Finally, Section \ref{sec:disc} concludes the paper. 

\section{Analytical confidence intervals for quantiles based on order statistics} \label{sec:AOS-CI}
In this section we briefly describe analytical confidence intervals (CIs) for quantiles based on order statistics to build intuition for the bootstrap proposed in this paper, and make the connection to the literature. This type of confidence interval has been proposed by many (see e.g. \citealp{Scheffe1945,gibbons2014, david2004order}) and there are several extensions and improvements of the standard solution \citep{nyblom1992,hutson1999}. Here we present the standard CIs. 

Let $Y_i$ be a random variable with cumulative distribution function $F$, and let $\mathbf{y}=(y_1,...,y_i,...,y_N)'$ be a sample of size $N$ obtained by sampling independently from a continuous distribution $F$. Let $0<q<1$ be the population quantile of interest, and $F^{-1}(q)$ be the population value at the quantile. The order-statistic CIs are based on the following simple reasoning. When one observation is sampled from the population, the probability that the observation is smaller than $F^{-1}(q)$ is equal to $q$, and the probability that the observation is larger than $F^{-1}(q)$ is equal to $1-q$. When $N$ observations are sampled from the population independently, this enables the following distribution-free CI. Select integers $r$ and $s$ such that $1\leq r< s \leq N$ and
\begin{equation}
    P(Y_{(r)}<F^{-1}(q)<Y_{(s)}) =\sum_{i=r}^{s-1}\binom{N}{i}q^i(1-q)^i\geq 1- \alpha,
\end{equation}
then the interval $(Y_{(r)}, Y_{(s)})$ is a $1- \alpha$ CI for $F^{-1}(q)$, where $Y_{(i)}$ refers to the $i$th order statistic. There is not a unique pair $(r, s)$ that satisfies the preceding equation. Additional restrictions can be added to find a unique pair, like assigning equal probability to each tail. See, e.g., \citet[p. 158]{gibbons2014} for details. Throughout this paper we will refer to the CI given above as analytical order-statistic-based CIs (AOS-CI). 

For the purposes of this paper, the most important aspect of the AOS-CI is that they use the fact that the distribution of order-statistic indexes is independent of the outcome data distribution. The binomial probabilities used above are valid due to the properties of quantile definition rather than properties of the outcome data. An important limitation of the AOS-CIs is that they are not applicable to difference-in-quantiles; the relation between the population quantiles and order statistics does not translate to the difference-in-quantiles and the difference in order statistics. Nevertheless, in this paper we use similar arguments in a bootstrap context to enable two-sample inference for quantiles. The following section gives a brief introduction to bootstrap inference.

\section{Poisson bootstrap for quantiles}\label{sec:intro_poi_boot}
The Poisson bootstrap \citep{hanley2006} works analogously to a standard non-parametric multinomial bootstrap, but it lets the number of observations in each bootstrap vary. Let again $Y_i$ be the random variable of interest and $y_i$ be an observation of that variable. Using the Poisson bootstrap, given a sample of size $N$, we independently generate $p_i^{(b)}\sim \poi(1)$ for all $i=1,...,N$. In each boostrap sample, $y_i$ is included $p_i$ times to form the bootstrap sample $\mathbf{y}^{(b)}$. We repeat the procedure $B\in \mathbbm{Z}^+$ number of times. The size of a given bootstrap sample is $\sum_{i=1}^Np_i^{(b)}$, which is equal to $N$ only in expectation. 
Let $y_{(i)}$ represent the $i$th order statistic in the sample $\mathbf{y}$. In this paper we are studying the quantile estimator of the form given in Definition \ref{def_estimator}. 
\begin{definition}\label{def_estimator}
Define the function $g$ as
$$g[q,N]=\begin{cases}  q(N+1) & \text{ if } q(N+1)\mod 1 = 0 \\
 (1-D)\left\lfloor {q(N+1)}\right\rfloor + D\left\lceil {q(N+1)}\right\rceil & \text{ if } q(N+1)\mod 1 \neq 0
\end{cases},$$ where $D\sim Ber(q(N+1)\mod 1)$. 
Define the sample quantile estimator of quantile $q$ as
\begin{equation*}
    \hat{\tau}_q = Y_{\left(g[q,N]\right)}.
\end{equation*}
\end{definition}
The function $g$ in Definition \ref{def_estimator} can be thought of as a stochastic rounding function. If $q(N+1)$ is an integer, the ordered observation $q(N+1)$ is the sample quantile. If $q(N+1)$ is not an integer, it is randomly rounded up or down with probability proportional to the remainder and the corresponding order statistic selected. The estimator is similar to most popular quantile estimators, but instead of weighting together the two closest observations when the quantile index is non-integer, it randomly selects one of them. This formulation of the quantile estimator implies that the estimate is always an observation from the original sample, which is key for the following results.

\subsection{Poisson bootstrap inference for quantiles through resampling}
In this section, standard resampling-based Poisson bootstrap for quantiles is presented, to build intuition for the proposed alternative presented in the following section. Even though bootstrap is never needed to construct CIs for the one-sample quantile (since AOS-CI can be used directly), we focus on the one-sample case here to build intuition. The two-sample difference-in-quantiles case is based on the same logic and is achieved by simple extensions presented in Section \ref{sec:two-sample-extension}. 

A standard Poisson bootstrap CI algorithm for $\hat{\tau}_q$ is given in Algorithm \ref{poi-bootstrap-one-sample}.  
\begin{algorithm}
\caption{Algorithm for Poisson Bootstrap confidence interval for a one-sample quantile.}\label{poi-bootstrap-one-sample}
\begin{enumerate}
    \item Generate $N$ $\poi(1)$ random variables $p^{(b)}_1,\dots,p^{(b)}_i,\dots,p^{(b)}_N$
    \item include each $y_i$ observation $p_i$ times and form the bootstrap sample outcome vector $\mathbf{y}^{(b)}$
    \item Calculate the sample estimate $\hat{\tau}^b_q=y_{(g[q,\sum_{i=1}^Np_i])}^{(b)}$
    \item Repeat steps 1--3 $B$ times
    \item Return the $\alpha/2$ and $1-\alpha/2$ quantiles of the distribution of $\hat{\tau}^b$ as the two-sided $(1-\alpha)100$\% confidence interval 
\end{enumerate}
\end{algorithm}
Algorithm \ref{poi-bootstrap-one-sample} requires that each bootstrap sample vector $\mathbf{y}^{(b)}$ is realized and ordered such that the order statistic that is the sample quantile estimate can be extracted. This is a memory and computationally intensive exercise. In the following section, we exploit theoretical properties of $\hat{\tau}_q$ to provide substantial simplifications to Algorithm \ref{poi-bootstrap-one-sample} that ultimately will enable bootstrap CIs for the difference-in-quantiles case.

\section{Poisson bootstrap CIs for $\hat{\tau}_q$ without resampling}\label{sec:resampling_free_boot}
The key insight that facilitates our approach is that the estimator $\hat{\tau}_q^{(b)}$ in Definition \ref{def_estimator} applied to the $b$th bootstrap sample only uses order statistics from the original sample. The process of drawing resampling frequencies $p_i^{(b)}$ and realizing the vector $\mathbf{y}^{(b)}$ is only necessary to find what original index maps to the sample quantile in the bootstrap sample. If we knew the distribution of the random variable describing which original index is the desired sample quantile in the bootstrap sample, we could simply generate indexes from that distribution and extract the corresponding original sample order statistic directly. 

We now aim to build intuition for this index distribution. Consider an example where we are interested in a confidence interval for the median in a sample of $N=10$ observations. According to the quantile estimator in Definition \ref{def_estimator}, the median is $y_{(5)}$ or $y_{(6)}$ with equal probability. It is intuitively apparent that the middle order statistics $y_{(4)},y_{(5)},y_{(6)}, y_{(7)}$ in the original sample are more likely to be medians in the bootstrap sample. The first and the last order statistics $y_{(1)},y_{(10)}$ have little chance of being the medians in the bootstrap sample. For example,, $y_{(1)}$ will be the median in a bootstrap sample only if $p_1\geq\sum_{i=2}^Np_i$ is satisfied. If, e.g., $p_1=2$ then the sum of all 9 remaining Poisson random variables ($p_2,\dots,p_{10}$) must be smaller than or equal to $p_1=2$, which is unlikely given that $\sum_{i=2}^Np_i\sim \poi(9)$. Following this logic, it is clear that order statistics that are closer in rank to the original sample quantile have higher probability of being observed as the sample quantile in a bootstrap sample. It is easy to simulate the distribution of what index in the original sample is observed as the desired quantile across bootstrap samples. Figure \ref{fig:simulation_1_index_dist} displays this distribution for the 10th percentile in $1,000,000$ Poisson bootstrap samples in a sample of size $N=2000$. 
\begin{figure}
    \centering
    \includegraphics{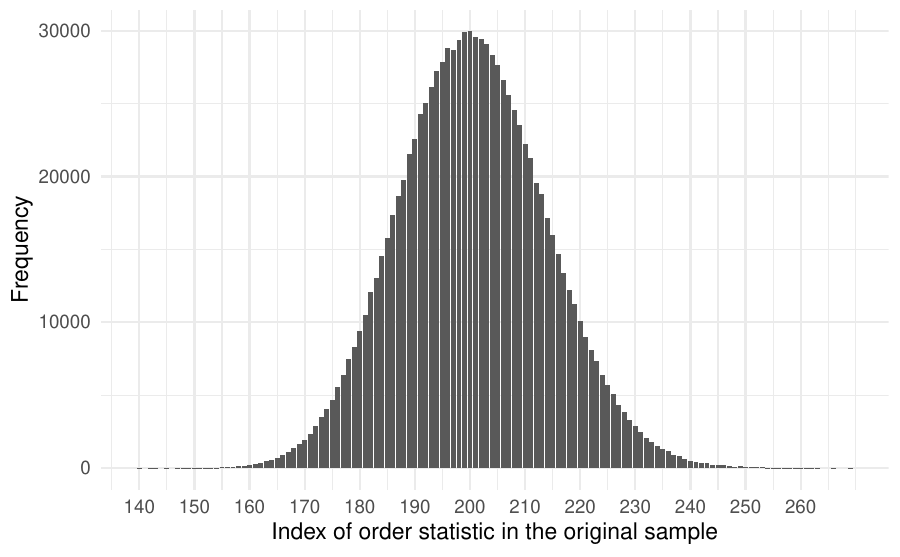}
    \caption{Distribution of index of the order statistics from a sample of $N=2000$ that became the 10th percentile over 1M Poisson bootstrap samples.}
    \label{fig:simulation_1_index_dist}
\end{figure}
If indexes could be generated directly from the index distribution, the quantile for a bootstrap sample could be obtained by simply generating an index and returning the corresponding order statistic from the original sample. In the following section, this index distribution is characterized mathematically.

\subsection{The probability distribution describing what index in the original sample is observed as the bootstrap sample quantile $q$ estimate $\hat{\tau}_q$}
Let $P_i\sim \poi(1)$ for $i=1, \dots, N$ be the frequencies used in the Poisson bootstrap. Denote $X_{<i}=\sum_{j=1}^{i-1}P_j$, $X_{>i}=\sum_{j=i+1}^{N}P_j$ and $S=X_{<i}+P_i+X_{>i}=\sum_{i=1}^NP_i$. By construction $X_{<i}\sim \poi(i-1)$ for $i>1$, $X_{>i}\sim \poi(N-i)$ for $i<N$, and $S\sim \poi(N)$. The following theorem establishes the distribution of $\psi$. 
\begin{theorem}\label{thm:index_dist}
Let $\psi\in \{1,\dots,N\}$ be the random variable that denotes what index of the original order statistics is observed as the bootstrap sample quantile $q$ estimate $\hat{\tau}_q$ (Definition \ref{def_estimator}) in a Poisson bootstrap sample. The probability mass function of $\psi$ is
\begin{multline*}
    p(\psi=i)=\\
    \sum_{n=0}^\infty p_{X_{<i},P_i,X_{>i}|S=n}\left(X_{<i}\leq q(n+1)-1, P_i>0, X_{>i}\leq (1-q)(n+1)-1|S= n\right)p_S(S=n)I(r=0)\\
    +r\sum_{n=0}^\infty p_{X_{<i},P_i,X_{>i}|S=n}\left(X_{<i}\leq q(n-1)-r, P_i>0, X_{>i}\leq (1-q)(n+1)+r-2|S=n\right)p_S\left(S=n\right)I(r\neq 0)\\
    +(1-r)\sum_{n=0}^\infty p_{X_{<i},P_i,X_{>i}|S}\left(X_{<i}\leq q(n+1)-r-1, P_i>0, X_{>i}\leq (1-q)(n+1)+r-1|S=n\right)p_S\left(S=n\right)I(r\neq 0),
\end{multline*}
where $r=q(n+1)\mod 1$.
\end{theorem}

\begin{proof}

Let $n$ be the realization of $S$, and $q(n+1)\mod 1 = r$. If $r =0$, then $q(n+1)$ is an integer, implying that for index $i$ to be the quantile the following must be satisfied
\begin{align*}
    X_{<i}&\leq q(n+1)-1\\
    P_i&> 0\\
    X_{>i}&\leq (1-q)(n+1)-1.
\end{align*}

If $r\neq 0$ the index selected will be $\lceil q(n+1)\rceil$ with probability $r$ and $\lfloor q(n+1)\rfloor$ with probability $1-r$. In the former case, the conditions that need to be satisfied for $i$ to be the index is
\begin{align*}
    X_{<i}&\leq q(n+1)-r\\
    P_i&> 0\\
    X_{>i}&\leq (1-q)(n+1)+r-2.
\end{align*}
When $q(n+1)$ is rounded down, the conditions are instead
\begin{align*}
    X_{<i}&\leq q(n+1)-r-1\\
    P_i&> 0\\
    X_{>i}&\leq (1-q)(n+1)+r-1.
\end{align*}
The expression in the theorem then follows from the law of total probability.
\end{proof}
While Theorem \ref{thm:index_dist} presents the distribution of $\psi$, it is not a tractable distribution that lends itself to being characterized easily. We will return to the matter of practical applications in Section \ref{sec:approximating_with_binom}. For the one-sample case, Theorem \ref{thm:index_dist} can be used to obtain Poisson bootstrap confidence intervals for the quantile analytically.
We formalize this result in Corollary \ref{clr:one-sample-ci}.
\begin{corollary}\label{clr:one-sample-ci}
Denote the lower and upper confidence interval bounds that result from Algorithm \ref{poi-bootstrap-one-sample} as $C^L_{\psi, \alpha/2}$ and $C^U_{\psi, 1-\alpha/2}$. Let $\psi$ be the random variable of the order-statistic index that becomes the quantile estimate in the bootstrap sample as defined in Theorem \ref{thm:index_dist}. Denote $i_L=\max_i \{i: P(\psi\leq i)\leq\alpha/2 \}$ and $i_U=\min_i \{i: P(\psi\geq i)\geq1-\alpha/2 \}.$ Then the confidence interval for the quantile estimator $\hat{\tau}_q$ given by $(Y_{(i_L)}, Y_{(i_U)})$ has coverage $\leq 1- \alpha$.
\end{corollary}
\begin{proof}
Since
\begin{align*}
   C^L_{\psi, \alpha/2} &\geq Y_{(i_L)}\\
   C^U_{\psi, \alpha/2} &\leq Y_{(i_U)},
\end{align*}
the results follows directly from Theorem \ref{thm:index_dist}.
\end{proof}
Interestingly, Corollary \ref{clr:one-sample-ci} gives CIs that are similar in construction to the AOS-CIs (Section \ref{sec:AOS-CI}), although derived based on two distinct approaches. In the following section we present an approximations of $p(\psi=i)$ that makes it easy to generate values that can in turn be used to enable difference-in-quantiles CIs.

\subsection{Approximating the index distribution to enable fast, resampling-free bootstrap inference for difference-in-quantiles}\label{sec:approximating_with_binom}
In this section, we propose an approximation of $p(\psi=i)$ to enable resampling-free difference-in-quantiles bootstrap CIs that are easy to implement. To motivate our approximation, Figure \ref{fig:approx_bin} shows two examples of index distributions with the probability mass function of the $\bin(N+1, q)$ distribution overlaid. The binomial distribution provides an impressive fit and the fit seems to improve with increasing $N$. We have found this surprisingly simple approximation to work incredibly well, and its simplicity means it is fast and easy to work with.
\begin{figure}
     \centering
     \begin{subfigure}[b]{0.45\textwidth}
         \centering
         \includegraphics[width=\textwidth]{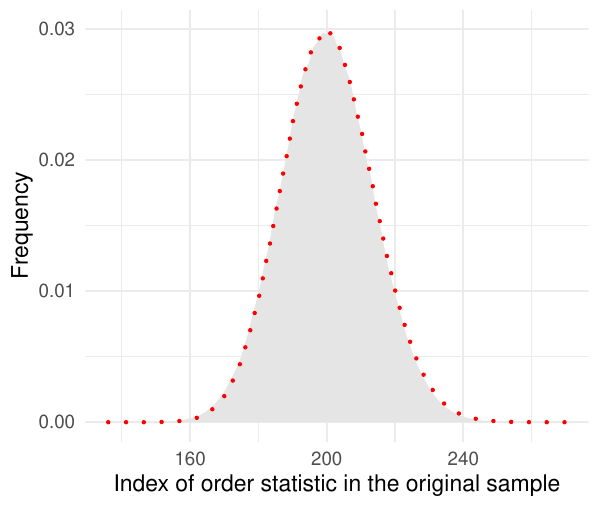}
         \caption{$N=2000, q=0.1$}
         \label{fig:approx_bin_1}
     \end{subfigure}
     \begin{subfigure}[b]{0.45\textwidth}
         \centering
         \includegraphics[width=\textwidth]{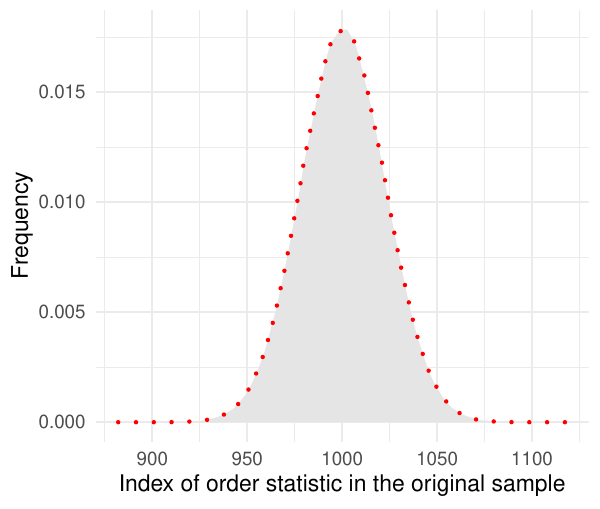}
         \caption{$N=2000, q=0.5$}
         \label{fig:approx_bin_2}
     \end{subfigure}
        \caption{Two examples of index distributions and their respective binomial approximations. The area and dotted curves show the estimated densities for the binomial approximation and index distribution, respectively.}
        \label{fig:approx_bin}
\end{figure}
We next state the use of this approximation as a conjecture, and then proceed to demonstrate its merit in Monte Carlo simulations.
\begin{conjecture}\label{thm:index_dist_approx_binom}
Let $X\sim \bin(N+1,q)$ and $G_N$ denote its cdf. Let also $H_N$ be the cdf of the index distribution as defined in Theorem \ref{thm:index_dist}. Then 
\begin{equation*}
    \sup_{x\in \{1, \dots, N\}}|G_N(x)-H_N(x)| \rightarrow 0 \text{ as } N\rightarrow \infty.
\end{equation*}
\end{conjecture}
Conjecture \ref{thm:index_dist_approx_binom} says that the bootstrap index distribution can be approximated with increased accuracy by a binomial distribution as the sample size increases, which is here defined as an asymptotically vanishing Kolmogorov-Smirnov distance between the distributions. 

If the binomial distribution is used to approximate the index distribution, the confidence interval in Corollary \ref{clr:one-sample-ci} coincides exactly with the AOS-CI using $\alpha/2$ in each tail of the index distribution. This implies that the coverage for the one-sample CI is exactly bounded by construction, but it says nothing about how similar $G_N(x)$ is to $H_N(x)$. However, the binomial-approximated index distribution and the binomial distribution used to derive AOS-CI are two quite distinct distributions. That is, they here happen to coincide exactly, but they describe fundamentally different things; the distribution of original-sample indexes of order statistics observed as quantiles in bootstrap samples versus the probability of a certain number of order statistics to be above or below the population quantile.

The distribution of original-sample indexes of order statistics observed as a given quantile in bootstrap samples is independent of the data-generating process as long as the outcome can be ordered. This means that Monte Carlo simulation can provide strong evidence that generalizes to all such data-generating processes. The setup of the simulation is the following. The number of bootstrap samples is $B=10^6$ and the sample size is set to $N\in\{100, 200, 500, 1000, 5000, 10000\}$. The quantile of interest is $q\in \{ 0.01, 0.1,0.25, 0.5\}$, which, due to symmetry, generalizes also to $q \in \{0.75, 0.9, 0.99\}$. For each combination of sample size and quantile, $10^6$ bootstrap samples are realized, and it is recorded which index from the ordered original sample that is observed as estimate of the quantile. The empirical distribution function of the indexes across the bootstrap samples are fitted. The Kolmogorov-Smirnov (KS) distance is calculated comparing the empirical bootstrap distribution to a $\bin(N+1, q)$ distribution.   

Figure \ref{fig:mc_ks} displays the Kolmogorov-Smirnov distance for each combination of quantile and sample size. In support of Conjecture \ref{thm:index_dist_approx_binom}, the KS distance is decreasing in sample size, indicating that the approximation of $\psi$ using the $\bin(N+1,q)$ distribution is improving as the sample size increases. Perhaps surprisingly, the approximation is not strictly improving as the quantile comes closer to 0.5. A likely explanation is that although the skewness and boundedness of the distribution is less heavy the closer the quantile is to 0.5, the variance in the index distribution also increases.    
\begin{figure}
    \centering
    \includegraphics{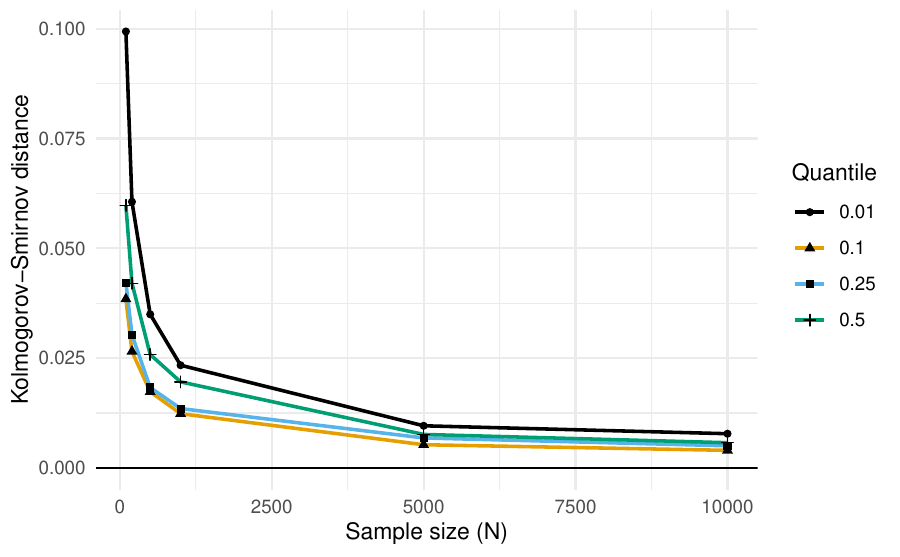}
    \caption{The Kolmogorov-Smirnov distance between the empirical index-distribution and $\bin(N+1,q)$, for quantiles 0.01, 0.1, 0.25, and 0.5 over sample sizes between 100 and 10000.}
    \label{fig:mc_ks}
\end{figure}

\section{Poisson Bootstrap CIs for difference-in-quantiles}\label{sec:two-sample-extension}
We now extend our approach to two-sample difference-in-quantiles inference. Assume that the the control and treatment groups are of sizes $N_c$ and $N_t$, respectively, such that the total sample size is $N_c+N_t$. Let the outcome of the control and treatment groups be denoted $\mathbf{y}_c=(y_{c,1},\dots,y_{c,i},\dots,y_{c,N_c})'$ and $\mathbf{y}_t=(y_{t,1},\dots,y_{t,j},\dots,y_{t,N_t})'$, respectively. Define the difference-in-quantile estimator $\hat{\delta}=\hat{\tau}^b_{t,q}-\hat{\tau}^b_{c,q}$, where subscripts $c$ and $t$ indicates the control and treatment groups, respectively. Algorithm \ref{poi-bootstrap-two-sample} defines the standard resampling-based algorithm for a Poisson bootstrap difference-in-quantile CI.
\begin{algorithm}
\caption{Algorithm for Poisson Bootstrap confidence interval for a two-sample difference-in-quantile $q$.}\label{poi-bootstrap-two-sample}
\begin{enumerate}
    \item Generate $N_c+N_t$ $\poi(1)$ random variables $p^{(b)}_{c,i}, i=1, \dots, N_c$ and $p^{(b)}_{t,j}, j=1, \dots, N_t$.
    \item include each $y_{c,i}$ observation $p_i$ times and each $y_{c,j}$ observation $p_j$ times to form the bootstrap sample outcome vectors $\mathbf{y}_c^{(b)}$ and $\mathbf{y}_t^{(b)}$.
    \item Calculate the sample estimate $\hat{\delta}^b=\hat{\tau}^b_{t,q}-\hat{\tau}^b_{c,q}$
    \item Repeat steps 1--3 $B$ times
    \item Return the $\alpha/2$ and $1-\alpha/2$ quantiles of the distribution of $\hat{\delta}^b$ as the two-sided $(1-\alpha)100$\% confidence interval. 
\end{enumerate}
\end{algorithm}
This algorithm requires generating $B$ $\poi(1)$ random variables, realize the samples and find the appropriate order statistic for each bootstrap sample, calculate the difference between treatment and control for each bootstrap sample, and finally find the quantiles of the distribution of differences. This implies a total complexity of order $O(NB)+O(B)+O(B)=O(NB)$. 

Using Theorem \ref{thm:index_dist}, it is straightforward to improve the efficiency of Algorithm \ref{poi-bootstrap-two-sample} for difference-in-quantile CIs. Here, we will utilize Conjecture \ref{thm:index_dist_approx_binom} directly to find practically applicable approximations. As before, exact analytical results can be obtained by replacing the binomial distribution with $p(\psi=i)$.

It is not possible to find a direct analogue to Corollary \ref{clr:one-sample-ci} for the two-sample difference-in-quantiles CI. While the within-sample distribution of indexes is independent of the outcome data, the distribution of the difference between samples (i.e., the difference-in-quantile estimate) is not. 
 Instead, Theorem \ref{thm:index_dist} together with Conjecture \ref{thm:index_dist_approx_binom} can be applied to generate $B$ bootstrap quantile estimates for each sample, i.e., $(\hat{\tau}^{(1)}_{q,t},...,\hat{\tau}^{(B)}_{q,t})$, and $(\hat{\tau}^{(1)}_{q,c},...,\hat{\tau}^{(B)}_{q,c})$ for treatment and control, respectively. The bootstrap distribution of the difference-in-quantiles can be directly obtained by simply taking the difference between these two vectors.     
 Let $\mathbf{a}[\mathbf{v}]$ denote extraction of elements from the vector $\mathbf{a}$ according to the vector of indexes $\mathbf{v}$ where elements in $\mathbf{v}$ are bounded between 1 and the length of $\mathbf{a}$. Algorithm \ref{poi-bootstrap-two-sample-fast} displays an efficient algorithm for obtaining CIs for the difference-in-quantiles.
\begin{algorithm}
\caption{Algorithm for Poisson Bootstrap confidence interval for a two-sample difference-in-quantile $q$.}\label{poi-bootstrap-two-sample-fast}
\begin{enumerate}
    \item Generate $B$ random numbers from $\bin(N_c+1,q)$ and $B$ random numbers from $\bin(N_t+1,q)$ and save them in two vectors $\mathbf{I}_c$ and $\mathbf{I}_t$, respectively. 
    \item Order the outcome vectors $\tilde{\mathbf{y}}_{c}=(y_{c,(1)},\dots,y_{c,(N_c)})'$ and $\tilde{\mathbf{y}}_{t,}=(y_{t,(1)},\dots,y_{c,(N_t)})'$.
    \item Calculate the vector of difference-in-quantiles as $\boldsymbol{\hat{\tau}}=\tilde{\mathbf{y}}_{t}[\mathbf{I}_t]-\tilde{\mathbf{y}}_{c}[\mathbf{I}_c]$
    \item Return the $\alpha/2$ and $1-\alpha/2$ quantiles of $\boldsymbol{\hat{\tau}}$ as the two-sided $(1-\alpha)100$\% confidence interval for the difference-in-quantiles.  
\end{enumerate}
\end{algorithm}
Algorithm \ref{poi-bootstrap-two-sample-fast} generates $2B$ binomial random numbers, sorts two vectors of lengths $N_c$ and
$N_t$, extracts 2B numbers from arrays and calculates the difference, and finally finds the quantiles of the distribution. This leads to an overall complexity of order $O(2B)+O(N_c\log(N_c))+O(N_t\log(N_t))+O(2B)+O(B)+O(B)=\max(B, N_c\log(N_c), N_t\log(N_t))$. Since $\log(N)<B$ for all relevant pairs of $N$ and $B$, Algorithm \ref{poi-bootstrap-two-sample-fast} has lower complexity than Algorithm \ref{poi-bootstrap-two-sample} in all reasonable applications.

\subsection{Monte Carlo simulations of the CI coverage for Algorithm \ref{poi-bootstrap-two-sample-fast}}\label{sec:monte_carlo}
In this section, the coverage of the confidence intervals resulting from algorithm \ref{poi-bootstrap-two-sample-fast} are studied using Monte Carlo simulation. The algorithms are implemented in \texttt{Julia} version 1.6.3 \citep{bezanson2017julia}, and the code for the algorithms and the Monte Carlo simulations can be found here \url{https://github.com/MSchultzberg/fast_quantile_bootstrap}.

The data-generating process is similar to the previous simulation. The number of Monte Carlo replications is $10^4$. For each replication, two samples of $N_t=N_c=10^5$, respectively, are generated from a standard normal distribution. The number of bootstrap samples for each Monte Carlo replication is $B=10^5$ and the two-sided 95\% confidence interval is returned. The study is repeated for the quantiles 0.01, 0.1, 0.25, and 0.5. The coverage rate is the proportion of the CIs that covered the true population difference-in-quantiles, i.e., zero. To quantify the error due to a finite number of Monte Carlo replications, the two-sided 95\% confidence intervals of the coverage rate (using standard normal approximation of the proportion) are again presented with the results. 

Table \ref{tab:mc-coverage-two-sample-fast} displays the results from the Monte Carlo simulation.
\begin{table}[ht]
\centering
\begin{tabular}{rlll}
\toprule
 &Empirical&\multicolumn{2}{c}{95\% CI}\\\cline{3-4}
$q$ & coverage & Lower & Upper \\ 
\midrule 
0.01 & 0.953 & 0.949 & 0.957 \\ 
  0.10 & 0.949 & 0.944 & 0.953 \\ 
  0.25 & 0.950 & 0.946 & 0.955 \\ 
  0.50 & 0.949 & 0.945 & 0.954 \\ 
   \hline
\end{tabular}
\caption{Empirical coverage rate for the confidence intervals produced by Algorithm \ref{poi-bootstrap-two-sample-fast} for the difference-in-quantiles for quantiles 0.01, 0.1, 0.25, and, 0.25, for sample size $10^5$ with $10^5$ bootstrap samples over 10000 replications.}\label{tab:mc-coverage-two-sample-fast}
\end{table}
Again, it is clear that the coverage is close to the intended 95\% for all quantiles, with no observable systematic deviations.  

\subsection{Time and memory simulation comparisons}
This section presents memory and time consumption comparisons to build intuition for the impact of the reduction in complexity enabled by Theorem \ref{thm:index_dist} and Conjecture \ref{thm:index_dist_approx_binom}. The comparisons are between Algorithm \ref{poi-bootstrap-two-sample} and \ref{poi-bootstrap-two-sample-fast} implemented in \texttt{Julia} version 1.6.3 \citep{bezanson2017julia} and benchmarked using the \texttt{BenchmarkTools} package \citep{BenchmarkTools2016}.

The setup for the comparison is the following. Two samples of floats are generated of size 1000 each. $B$ is set to 10000. The setup is selected to enable 100 evaluations of Algorithm \ref{poi-bootstrap-two-sample} within around 200 seconds on a local machine. The results are displayed in Table \ref{tab:time_and_speed_sim}. 
\begin{table}[hbt!]
    \centering
    \begin{tabular}{r|llll}
         &Min time&Median time &Max time & Memory usage  \\
         \hline
         Algorithm \ref{poi-bootstrap-two-sample}& 1726 ms&1821ms &1902ms&2.39 GiB\\
         Algorithm \ref{poi-bootstrap-two-sample-fast} & 2.055ms & 2.214ms&3.502ms & 407.08 KiB \\
         \hline
    \end{tabular}
    \caption{Time and mempry consumption comparison between a standard Poisson bootstrap algorithm (Algorithm \ref{poi-bootstrap-two-sample}) for a difference-in-quantiles CIs and the corresponding proposed binomial-approximated Poisson bootstrap algorithm (Algorithm \ref{poi-bootstrap-two-sample-fast}).}
    \label{tab:time_and_speed_sim}
\end{table}
Clearly, Algorithm \ref{poi-bootstrap-two-sample-fast} outperforms Algorithm \ref{poi-bootstrap-two-sample} both in terms of memory and speed already for small samples and moderately small $B$. The results presented in Table \ref{tab:time_and_speed_sim}, together with the simulation results in Section \ref{sec:MC-coverage}, establishes the utility and practical implications that follows from the theoretical results.  

\section{Discussion and conclusion}\label{sec:disc}
In this paper we exploit the properties of quantile estimators coupled with a Poisson bootstrap sampling scheme to derive computationally simple bootstrap inference algorithms to make difference-in-quantiles inference feasible in large-scale experimentation. It turns out that for the quantile estimator we employ, no resampling is necessary. Instead, the theoretical distribution of the indexes of order statistics in the original sample that are observed as the quantile estimate in the bootstrap sample can be derived and used directly. The traditional algorithm is built around generating Poisson random variables for each observation in the original sample, realizing the bootstrap sample, and selecting the order statistic in the bootstrap sample that is closest to the desired quantile. In this paper we show that it is possible, due to the known properties of the Poisson bootstrap sampling mechanism, to describe probabilistically which order statistic in the original sample is observed as the desired quantile in a bootstrap sample. This effectively bypasses the need for realizing each bootstrap sample. In addition, we show that the index distribution, which has an analytically intractable exact distribution, is well approximated by a binomial distribution that simplifies implementation dramatically. 

Together our findings enables bootstrap inference for quantiles and difference-in-quantiles in large-scale experiments without the need for intricate parallelization implementations. In fact, a simple SQL query coupled with a Python or R notebook is sufficient for even the largest experiments with even hundreds of millions of users. We hope that this will enable fast and robust inference for quantiles for many large-scale experimenters. 

We leave for future research to study the properties of $p(\psi=i)$ in more detail. If the distribution could be exactly or approximately characterized in a manner that made generation of random numbers straightforward, that would open up faster bootstrap algorithms also for smaller sample sizes. This might also enable proving or narrowing Conjecture \ref{thm:index_dist_approx_binom}.

\pagebreak
\bibliography{library}
\end{document}